\documentclass[11pt]{article}
\usepackage[margin=1in]{geometry}
\usepackage{libertine}

\usepackage{datetime}
\usepackage{multirow}
\usepackage{booktabs}
\usepackage{enumitem}
\usepackage{wrapfig}
\usepackage{subcaption}
\usepackage[T1]{fontenc}
\usepackage[utf8]{inputenc}
\usepackage{pgfplots}
\usepackage{amsmath, amssymb}
\usepackage{amsthm}
\usepackage{color}
\usepackage{cases}
\usepackage{xparse}
\usepackage{xargs}
\usepackage{appendix}
\usepackage[boxed,commentsnumbered,noend]{algorithm2e}
\usepackage{algpseudocode}
\usepackage{verbatim, xspace}
\usepackage{tikz}
\usepackage{mathtools}

\usepackage{hyperref}
\usepackage{pgfplots}

\usepackage{xcolor}

\newcommand{\citet}[1]{\cite{#1}}
\usepackage[colorinlistoftodos,prependcaption,textsize=tiny]{todonotes}
\newcommandx{\unsure}[2][1=]{\todo[linecolor=green,backgroundcolor=green!25,bordercolor=green,#1]{\normalsize #2}}
\newcommandx{\improvement}[2][1=]{\todo[inline,linecolor=blue,backgroundcolor=blue!05,bordercolor=blue,#1]{\normalsize #2}}
\newcommandx{\info}[2][1=]{\todo[linecolor=yellow,backgroundcolor=yellow!25,bordercolor=yellow,#1]{#2}}
\newcommandx{\floatmodel}[2][1=]{\todo[inline,linecolor=red,backgroundcolor=yellow!25,bordercolor=yellow,#1]{#2}}
\newcommandx{\thiswillnotshow}[2][1=]{\todo[disable,#1]{#2}}
\newcommandx{\karol}[2][1=]{\todo[inline,linecolor=blue,backgroundcolor=blue!25,bordercolor=blue,caption={\normalsize \textbf{Karol}},#1]{\normalsize #2}}
\newcommandx{\jana}[2][1=]{\todo[inline,linecolor=red,backgroundcolor=red!25,bordercolor=red,caption={\normalsize
\textbf{jana}},#1]{\normalsize #2}}
\newcommandx{\michal}[2][1=]{\todo[inline,linecolor=gray,backgroundcolor=red!25,bordercolor=red,caption={\normalsize \textbf{Micha\l{}}},#1]{\normalsize #2}}

\newtheorem{theorem}{Theorem}
\newtheorem{definition}[theorem]{Definition}
\newtheorem{lemma}[theorem]{Lemma}

\newtheorem{claim}[theorem]{Claim}

\newtheorem{problem}{Open Question}
\newtheorem*{maingoal*}{Main Question}

\numberwithin{theorem}{section}

\newcommand{\eps}{\varepsilon}

\newcommand{\Oh}{\mathcal{O}}

\newcommand{\nat}{\mathbb{N}}

\newcommand{\Rr}{\mathcal{R}}
\newcommand{\Dd}{\mathcal{D}}
\newcommand{\Ss}{\mathcal{S}}
\newcommand{\Ii}{\mathcal{I}}

\newcommand{\Tt}{\mathcal{T}}
\newcommand{\Mm}{\mathcal{M}}
\newcommand{\Nn}{\mathcal{N}}

\newcommand{\Gg}{\mathcal{G}}

\newcommand{\Hh}{\mathcal{H}}
\newcommand{\real}{\mathbb{R}}

\newcommand{\opt}{\mathrm{opt}}

\renewcommand{\leq}{\leqslant}
\renewcommand{\geq}{\geqslant}
\renewcommand{\le}{\leqslant}
\renewcommand{\ge}{\geqslant}

\renewcommand{\setminus}{-}

\title{Parameterized Approximation for Maximum Weight Independent~Set of Rectangles and Segments}
\date{}

\author{
    Jana Cslovjecsek\footnote{EPFL, Switzerland, \textsf{jana.cslovjecsek@epfl.ch}.}
    \and
    Micha\l{} Pilipczuk\footnote{Institute of Informatics, University of
    Warsaw, Poland, \textsf{michal.pilipczuk@mimuw.edu.pl}. This work is a part of
    the project BOBR that has received funding from the European
    Research Council (ERC) under the European Union's Horizon 2020 research and
    innovation programme (grant agreement No 948057).}
    \and
    Karol W\k{e}grzycki\footnote{Saarland University and Max Planck Institute for Informatics,
        Saarbr\"ucken, Germany, \textsf{wegrzycki@cs.uni-saarland.de}. 
    This work is part of the project TIPEA that has
    received funding from the European Research Council (ERC) under the European Union's Horizon
    2020 research and innovation programme (grant agreement No 850979).}
}

\begin{document}

\maketitle

\thispagestyle{empty}
\begin{abstract}

In the \textsc{Maximum Weight Independent Set of Rectangles} problem
(\textsc{MWISR}) we are given a weighted set of $n$ axis-parallel rectangles in
the plane. The task is to find a subset of pairwise non-overlapping rectangles
with the maximum possible total weight.  This problem is NP-hard and the
best-known polynomial-time approximation algorithm, due to by Chalermsook and
Walczak (SODA~2021), achieves approximation factor $\Oh(\log\log n )$. While in
the unweighted setting, constant factor approximation algorithms are known, due
to Mitchell (FOCS 2021) and to Gálvez et al. (SODA 2022), it remains open to
extend these techniques to the weighted setting.

In this paper, we consider \textsc{MWISR} through the lens of parameterized
approximation. Grandoni et~al. (ESA 2019) gave a $(1-\eps)$-approximation
algorithm with running time $k^{\Oh(k/\eps^8)} n^{\Oh(1/\eps^8)}$ time, where
$k$ is the number of rectangles in an optimum solution. Unfortunately, their
algorithm works only in the unweighted setting and they left it as an open
problem to give a parameterized approximation scheme in the weighted setting.

Our contribution is a partial answer to the open question of Grandoni et al.
(ESA 2019). We give a parameterized approximation algorithm for \textsc{MWISR}
that given a parameter $k \in \nat$, finds a set of non-overlapping rectangles
of weight at least $(1-\eps) \opt_k$ in $2^{\Oh(k \log(k/\eps))}
n^{\Oh(1/\eps)}$ time, where $\opt_k$ is the maximum weight of a solution of
cardinality at most $k$.  Note that thus, our algorithm may return a solution
consisting of more than $k$ rectangles. To complement this apparent weakness, we
also propose a parameterized approximation scheme with running time $2^{\Oh(k^2
\log(k/\eps))} n^{\Oh(1)}$  that finds a solution with cardinality at most $k$
and total weight at least $(1-\eps)\opt_k$ for the special case of axis-parallel
segments.

\end{abstract}

\begin{picture}(0,0)
\put(462,-180)
{\hbox{\includegraphics[width=40px]{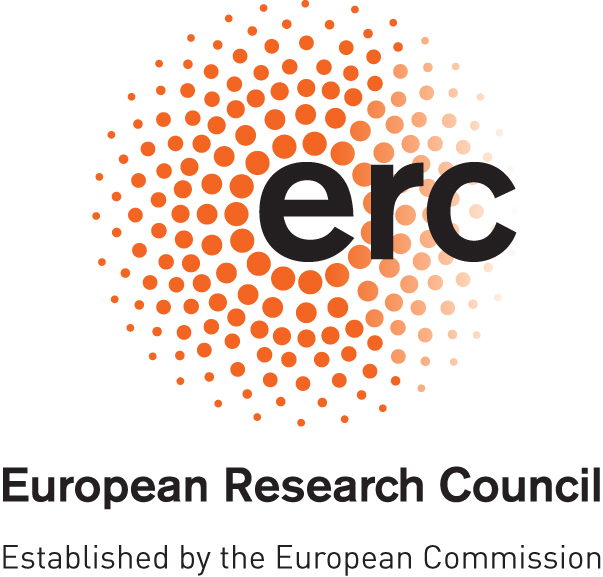}}}
\put(452,-240)
{\hbox{\includegraphics[width=60px]{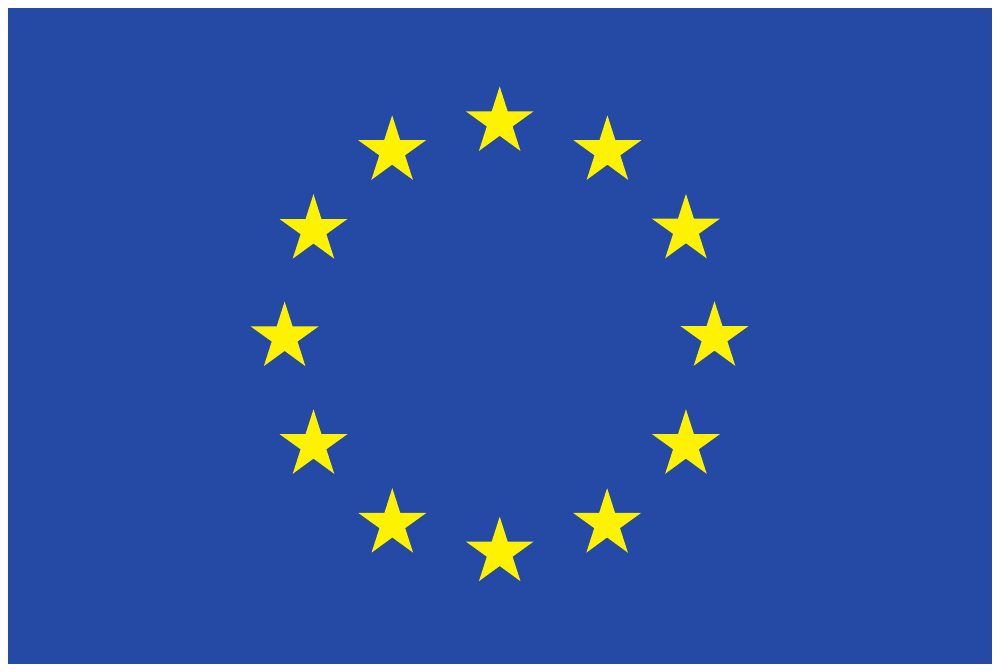}}}
\end{picture}

\clearpage
\setcounter{page}{1}

\section{Introduction}\label{sec:intro}

In the field of parameterized complexity the goal is to design an algorithm that
is efficient not only in terms of the input size, but also in terms of auxiliary
\emph{parameters}. On the other end of the spectrum, in the field of
approximation algorithms the goal is to design an algorithm that returns a
solution that is only slightly worse than the optimum one. These two notions are
traditional frameworks to deal with NP-hard problems. Recently, researchers
started to combine the two concepts and try to design approximation algorithms
that run in parameterized time. Ideally, given $\eps > 0$ and a parameter $k \in
\nat$, for example the size of the desired solution, one seeks an algorithm with
running time of the form $f(k,\eps) n^{g(\eps)}$ for some functions $f(k,\eps)$
and $g(\eps)$, which returns a $(1+\eps)$-approximate solution.  Such an
algorithm is called \emph{parameterized approximation scheme} (PAS).

In this paper we continue this line of work and apply it to a fundamental
geometric packing problem. In the \textsc{maximum weight independent set of
rectangles} (\textsc{MWISR}) problem we are given a set $\Dd$ consisting of $n$
axis-parallel rectangles in the plane alongside with a weight function $\omega
\colon \Dd \rightarrow \real$. Each rectangle $R \in \Dd$ is a closed set of
points $[x_1,x_2] \times [y_1,y_2]$ fully characterized by the positions of its four
corners. A feasible solution $\Ss \subseteq \Dd$ to the \textsc{MWISR} problem
consists of rectangles that are pairwise disjoint, i.e., for any two different
$R, R' \in \Ss$ we have $R \cap R' = \emptyset$; we also call such a solution an
{\em{independent set}}. The objective is to find a feasible solution of maximum
total weight.  In this paper, we consider a parameterized setting of the
problem. We use parameter $k \in \nat$ to denote the \emph{cardinality} of the
solution. Then $\opt_k(\Dd)$ denotes the maximum possible weight of an
independent set in $\Dd$ whose cardinality is at most $k$.

\textsc{MWISR} is a fundamental problem in geometric optimization. It naturally
arises in various applications, such as map
labeling~\cite{applications1,applications2}, data mining~\cite{data-mining},
routing~\cite{routing}, or unsplittable flow routing~\cite{bonsma}.
\textsc{MWISR} is well-known to be NP-hard~\cite{FOWLER1981133}, and it admits a
QPTAS~\cite{adamaszek13}. The currently best approximation factor achievable in
polynomial time is $\Oh(\log\log(n))$~\cite{parinya21}. From the parameterized
perspective, it is known that the problem is $\text{W}[1]$-hard when
parameterized by $k$, the number of rectangles in the solution, even in the
unweighted setting and when all the rectangles are squares~\cite{marx05}.
Therefore, it is unlikely that there is an exact algorithm with a running time of the
form $f(k) n^{\Oh(1)}$, even in this restricted setting. In particular, this
also excludes any $(1-\eps)$-approximation algorithm running in $f(\eps)
n^{\Oh(1)}$ time~\cite{bazgan,cesati}. We note that in the case of weighted
squares, there is a PTAS with running time of the form
$n^{g(\eps)}$~\cite{erlebach05}.

Approximation of \textsc{MWISR} becomes much easier in the unweighted setting.
With this restriction, even constant factor approximation algorithms for
\textsc{MWISR} are known~\cite{mitchell21,galvez22}, and there is a QPTAS with a
better running time~\cite{chuzhoy16}.  Grandoni et al.~\cite{esa19} were the
first to consider parameterized approximation for the \textsc{MWISR} problem,
although in the unweighted setting. They gave a parameterized approximation
scheme for unweighted \textsc{MWISR} running in
$k^{\Oh(k/\eps^8)}n^{\Oh(1/\eps^8)}$ time. As an open problem, they asked if one
can also design a PAS in the weighted setting.

\begin{problem}[\cite{esa19}]
    Does {\sc{Maximum Independent Set of Rectangles}} admit a parameterized approximation scheme in the weighted setting?
\end{problem}

\paragraph*{Our contribution.}
In this paper we partially answer the open question of Grandoni et
al.~\cite{esa19} by proving the following result:

\begin{theorem}\label{thm:main-rect}
    Suppose $\Dd$ is a set of axis-parallel rectangles in the plane with positive
    weights. Then given~$k$~and $\eps>0$, one can in $2^{\Oh(k \log(k/\eps))}
    |\Dd|^{\Oh(1/\eps)}$ time find an independent set in $\Dd$ of weight at least
    $(1-\eps)\opt_k(\Dd)$.    
\end{theorem}

Note that there is a caveat in the formulation above: the returned solution may
actually be of cardinality larger than $k$, but there is a guarantee that it
will be an independent set. This is what we mean by a ``partial'' resolution of
the question of Grandoni et al.~\cite{esa19}: ideally, we would like the
algorithm to return a solution of weight at least $(1-\eps)\opt_k(\Dd)$
{\em{and}} of cardinality at most $k$. At this point we are able to give such an
algorithm only in the restricted case of axis-parallel segments (see
Theorem~\ref{thm:main-segm} below), but let us postpone this discussion till
later and focus now on Theorem~\ref{thm:main-rect}.  Observe here that the issue
with solutions of cardinality larger than $k$ becomes immaterial in the
unweighted case, hence Theorem~\ref{thm:main-rect} applied to the unweighted
setting solves the problem considered by Grandoni et al.~\cite{esa19} and
actually improves the running time of their algorithm. 

We now briefly describe the technical ideas behind Theorem~\ref{thm:main-rect}.
Similarly to Grandoni et al.~\cite{esa19}, the starting point is a
polynomial-time construction of a grid such that each rectangle in $\Dd$
contains at least one gridpoint. However, in order to take care of the weights,
our grid is of size $(2k^2/\eps) \times (2k^2/\eps)$. Moreover, already in this
step we may return an independent set of weight at least $(1-\eps)\opt_k$ that
consists of more than $k$ rectangles. This is the only step where the algorithm
may return more than $k$~rectangles.

After this step, the similarities to the algorithm of Grandoni et
al.~\cite{esa19} end. We introduce the notion of the \emph{combinatorial type} of
a solution. This is simply a mapping from each rectangle in the solution to the set of
all  gridpoints contained in it. Observe that since the size of the grid is
bounded by a function of $k$ and $\eps$, we can afford to guess (by branching
into all possibilities) the combinatorial type of an optimum solution in
$f(k,\eps)$ time. Notice that there may be many different rectangles matching
the type of a rectangle from the optimum solution. However, it is possible that
such a rectangle overlaps with the neighboring rectangles (and violates
independence). Therefore, we need constraints that prevent rectangles from
overlapping.  For this, we construct an instance of \textsc{Arity-$2$ Valued
Constraint Satisfaction Problem} ({\sc{2-VCSP}}) based on the guessed
combinatorial type.

Next, we observe that this instance induces a graph that is ``almost'' planar,
hence we may apply a variant of Baker's shifting technique~\cite{baker94}. This
allows us to divide the instance into many independent instances of
{\sc{2-VCSP}} while removing only $\eps \cdot \opt_k$ weight from the optimum
solution. Moreover, each of these independent instances induces a graph of
bounded treewidth, and hence can be solved exactly in $|\Dd|^{\Oh(1/\eps)}$
time. This concludes a short sketch of our approach.

Let us return to the apparent issue that our algorithm may return a solution of
cardinality larger than $k$. This may happen in the very first step of the
procedure, during the construction of the grid. By employing a completely
different technique, we can circumvent this problem in the restricted setting of
axis-parallel segments and prove the following result.

\begin{theorem}\label{thm:main-segm}
    Suppose $\Dd$ is a set of axis-parallel segments in the plane with positive
    weights. Then given~$k$ and $\eps>0$, one can in
    $2^{\Oh(k^2\log(k/\eps))}|\Dd|^{\Oh(1)}$ time find an independent set in $\Dd$
    of cardinality at most $k$ and weight at least $(1-\eps)\opt_k(\Dd)$.
\end{theorem}

K{\'{a}}ra and Kratochv{\'{\i}}l~\cite{kratochvil06} and Marx~\cite{marx06}
independently observed that the problem of finding a maximum cardinality
independent set of axis-parallel segments admits an fpt algorithm. However,
their approach heavily relies on the fact that the problem is unweighted. In
this setting our approach is different: In fact, we show that finding a maximum
weight independent set of axis-parallel segments admits an algorithm with
running time $W^{\Oh(k^2)} |\Dd|^{\Oh(1)}$, where $W$ is the number of distinct
weights present among the segments.

We proceed with an outline of the proof of Theorem~\ref{thm:main-segm} and highlight
some technical ideas.  First, we modify the instance so that the number of
different weights is bounded.  This is done through guessing the largest weight
of a rectangle in an optimum solution and rounding the weights down.  This is
the only place where we lose accuracy on the optimal solution.  In other words,
the algorithm is fixed-parameter tractable in $k$ and the number of distinct
weights $W = (k/\eps)^{\Oh(1)}$.

With this assumption, we then construct a grid with $\Oh(k^2)$ lines hitting every segment of the instance.
We say that the grid is \emph{nice} with respect to a segment $I$, if $I$ contains a grid point;
equivalently, $I$ is nice if it is hit by two orthogonal lines of the grid.
Observe that the constructed grid is not necessarily nice for every segment of the instance.
We adapt the previously introduced notion of the \emph{combinatorial type} in order to also accommodate segments which do not contain a grid point.
This is done by mapping the segment to its four neighboring grid lines instead of the grid points contained inside the segment.
Further, the weight of the segment is added to its combinatorial type.
Similarly to before, the combinatorial type of a segment only depends on the grid size and the number of distinct weights.
This allows to guess (by branching into all possibilities) the combinatorial type of the optimum solution $\Ss$ in $k^{\Oh(1)}\cdot W$ time.

The goal is to construct a grid which is nice with respect to all segments of an optimal solution $\Ss$.
For this, we start by guessing the combinatorial type of all nice segments of an optimal solution $\Ss$.
Then, we incrementally guess the combinatorial type of the next heaviest segment in $\Ss$ for which the grid is not yet nice.
For each such combinatorial type, we find all possible candidate segments and add at most $k$ lines to the grid $G$.
This ensures that a correct candidate segments is hit both directions.
Repeating this procedure at most $k$ times we end up with a grid which is nice with respect to all the segments of $\Ss$.

Given such a grid, it remains to guess the combinatorial type of all segments in $\Ss$ and solve resulting instance.
This can be done either greedily or by observing that the problem can be modeled as a $2$-CSP instance whose constraint graph is a union of paths.
Both these cases work due to the fact that the segments only interact with each other when they lie on the same grid line.

\newcommand{\dist}{\mathrm{dist}}
\newcommand{\bag}{\mathsf{bag}}

\section{Preliminaries}

Through the paper, we silently assume that every weight is positive (because we
work on maximization, we can simply disregard objects of non-positive weights).
A {\em{tree decomposition}} of a graph $H$ is a tree $T$ together with a function $\bag$ that maps nodes of $T$ to subsets of vertices of $H$, called {\em{bags}}. The following conditions must be satisfied:
\begin{itemize}[nosep]
 \item 
for every vertex $u$ of $H$, the nodes of $T$ whose bags contain $u$ must form a connected, nonempty subtree of $T$; and
\item
for every edge $uv$ of $H$, there must exist a node of $T$ whose bag contains both $u$ and $v$.
\end{itemize}
The {\em{width}} of a tree decomposition $(T,\bag)$ is $\max_{x\in V(T)}|\bag(x)|-1$. The {\em{treewidth}} of $H$ is the minimum possible width of a tree decomposition of $H$. By $\dist_H(u,v)$ we mean the distance between vertices $u$ and $v$ in a graph $H$.

We use well-known results about {\sc{Arity-$2$ Valued Constraint
Satisfaction Problems}} ({\sc{2-VCSP}}s). An instance of {\sc{2-VCSP}} is a finite set of
\emph{variables} $X$, a domain $D_x$ for each variable $x\in X$, and a
set $C$ of \emph{constraints}. Each constraint $c\in C$ binds an ordered pair $\textbf{x}_c\in X\times X$ of variables in $X$ (not necessarily distinct) and is given by a mapping $f_c \colon
\prod_{x\in \textbf{x}_c} D_x \rightarrow \real$.
We assume that $f_c$ is given as the set of pairs
$\{(\textbf{d},f_c(\textbf{d})) \colon \textbf{d} \in \prod_{x\in \textbf{x}_c} D_x\}$. The goal is to compute the maximum value of the function $f(\textbf{u}) \coloneqq
\sum_{c\in C} f_c(\textbf{u}|_{\textbf{x}_c})$, over all possible assignments $\textbf{u}\in \prod_{x\in X} D_x$ of
values in respective domains to variables in $X$. The value $f(\mathbf{u})$ will be called the {\em{revenue}} of the assignment $\mathbf{u}$.

Observe that each instance of {\sc{2-VCSP}} induces an undirected graph, called the {\em{Gaifman graph}}: the vertex set is the set of variables~$X$, and for every pair of distinct variables $x,y\in X$, there is an edge $xy$ if and only if there is a constraint $c\in C$ such that $\textbf{x}_c=(x,y)$. Given a class $\Hh$ of graphs,
we can define a restriction of {\sc{2-VCSP}} to $\Hh$ by focusing only on instances whose Gaifman
graph is in $\Hh$. In this paper we focus only on instances of {\sc{2-VCSP}} where
the Gaifman graph has bounded treewidth. In this setting, it is well-known that a
standard dynamic programming solves {\sc{2-VCSP}}
efficiently~\cite{freuder90}\footnote{Freuder~\cite{freuder90} actually 
considered only the unweighted 2-CSP, however, as pointed out in e.g.,
\cite{zivny1,zivny2}, this dynamic-programming approach can be adapted to the weighted setting. Also, Freuder assumes that a suitable tree decomposition is given on input. Such a tree decomposition can be provided within the stated time complexity by, for instance, the $4$-approximation algorithm of Robertson and Seymour~\cite{RobertsonS95b}.}.

\begin{theorem}[\cite{freuder90}]
    \label{thm:vcsp}
    {\sc{2-VCSP}} can be solved in time $\Delta^{\Oh(t)}\cdot |X|^{\Oh(1)}$ when the
    Gaifman graph has treewidth at most $t$ and all domains are of size at most $\Delta$.
\end{theorem}

In Section~\ref{sec:segm} we also use standard {\sc{$2$-CSP}}s. These can be modeled by {\sc{$2$-VCSP}}s where all the constraints are hard: revenue functions $f_c$ assign only values $0$ (the constraint is satisfied) or $-\infty$ (the constraint is not satisfied). The task is to find a variable assignment that satisfies all constraints, that is, yields revenue~$0$.

\newcommand{\Hb}{{H^\bullet}}

\section{Axis-parallel rectangles}\label{sec:rect}

In this section we prove Theorem~\ref{thm:main-rect}. Therefore, we fix the given set $\Dd$ of weighted axis-parallel rectangles. For a rectangle $R \in \Dd$, the weight of $R$ is
$\omega(R) \in \real$.
By $\opt_k(\Dd)$ we denote the maximum possible weight of a set consisting of at most $k$ disjoint rectangles in $\Dd$. We also fix any {\em{optimum solution}}, that is, a set
$\Ss \subseteq \Dd$ of cardinality at most $k$ satisfying $\omega(\Ss)=\opt_k(\Dd)$.

We start with a simple preprocessing on $\Dd$. First, we guess a rectangle $R_{\max} \in \Ss$ with maximum weight among all rectangles of $\Ss$. This can be done with an extra overhead of $\Oh(n)$ in
the running time. Observe that $\omega(R_{\max}) \ge \opt_k(\Dd)/k$. Further, we
remove from $\Dd$ every rectangle of weight larger than $ \omega(R_{\max})$ and
every rectangle of weight not exceeding $\eps \omega(R_{\max})/k$; let the obtained instance be $\Dd'$. Observe that this operation does not decrease
the optimum significantly, as none of the former rectangles and at most $k$ of the latter rectangles could be used in $\Ss$.
More precisely, we have
\begin{displaymath}
    \opt_k(\Dd') \ge \opt_k(\Dd) - k \cdot \frac{\eps\cdot \omega(R_{\max})}{k} \ge (1-\eps) \opt_k(\Dd).
\end{displaymath}
After this preprocessing, the optimum decreased by at most
$\eps \cdot \opt_k(\Dd)$. This concludes the description of preprocessing. From now
on, we silently assume that our instance is~$\Dd \coloneqq \Dd'$.

\subsection{Constructing a grid}

\newcommand{\pnts}{\mathsf{points}}
 
We first introduce relevant terminology.
 
\begin{definition}
    A {\em{grid}} is a finite set of horizontal and vertical lines in the plane.  The
    \emph{size} $|G|$ of a grid $G$ is the total number of
     lines it contains. A {\em{grid point}} of $G$ is the intersection of a horizontal and a vertical line of~$G$. The set of grid points of $G$ is denoted by $\pnts(G)$. The lines of the grid divide the plane into {\em{grid cells}}. Thus, each grid cell is a rectangle, possibly with one or two sides extending to infinity, and at most four {\em{corners}}: the grid points lying on its boundary.
     
    A grid $G$ is \emph{good} for a set of axis-parallel rectangles $\Dd$, if for
    every rectangle $R \in \Dd$ there is a grid point of $G$ contained in $R$.
\end{definition}

As mentioned in Section~\ref{sec:intro}, our search for an optimal solution pivots around a bounded size grid that is good for the optimum solution. The construction of this grid is encapsulated in the following lemma.

\begin{lemma}
    \label{lem:step1}
    Suppose we are given a set $\Dd$ of weighted axis-parallel rectangles and let $\Delta(\Dd)$ be the ratio between lowest and highest weight in $\Dd$.
    Then, supposing $\Delta(\Dd)\ge \eps/k$ for some $\eps>0$, one can, in polynomial time, either
    \begin{itemize}[nosep]
        \item compute a grid $G$ of size $|G| \le \frac{2k^2}{\eps}$ that is good for $\Ss$, or
        \item return an independent set $\Ii \subseteq \Dd$ with  $\omega(\Ii) \ge \mathrm{opt}_k(\Dd)$.
    \end{itemize} 
\end{lemma}

\begin{proof}
    We construct the grid $G$ by first constructing the vertical lines of
    $G$, and then with basically the same procedure we add the horizontal lines of $G$. For the construction of the vertical lines, we iteratively pick vertically disjoint rectangles in a greedy
    fashion. For every rectangle $R \in \Dd$, select a point $p_R\in R$ very close to the top-right corner of $R$. We start with $\Dd_1 \coloneqq \Dd$.  In iteration
    $i \in \nat$, we select a rectangle $R_i^{\mathrm{ver}}\in \Dd_i$ for which $p_i\coloneqq p_{R_i^{\mathrm{ver}}}$ is the leftmost among rectangles of $\Dd_i$. (In case of ties, select any of the tying rectangles.) Then, add the vertical line
    $\ell^{\mathrm{ver}}_i$ which contains $p_i$ to the grid. Next, delete every rectangle from $\Dd_i$
    intersecting $\ell^{\mathrm{ver}}_i$, thus obtaining the next set $\Dd_{i+1}$. We repeat this procedure until
    no more rectangles are left in $\Dd_i$.
    To finish the
    construction of $G$, repeat the above algorithm in the orthogonal direction, thus selecting vertically disjoint rectangles $R_i^{\mathrm{hor}}$ and adding to $G$ horizontal lines $\ell_i^{\mathrm{hor}}$. This concludes the construction of $G$; see
    Figure~\ref{fig:grid} for an illustration.
    
    \begin{figure}[ht!]
        \centering
        \includegraphics[width=0.4\textwidth]{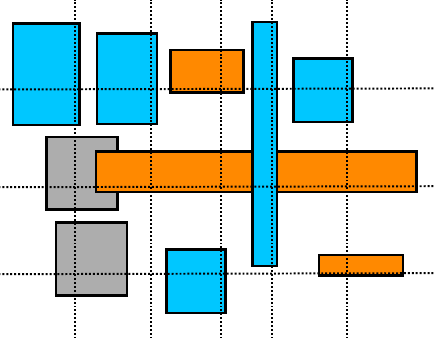}
        \caption{The grid constructed after applying the greedy procedure.
        Rectangles $R^{\mathrm{ver}}_i$ are
    blue-filled, and rectangles $R^{\mathrm{hor}}_i$ are orange-filled. Observe that every rectangle is hit by at least one grid point.}
        \label{fig:grid}
    \end{figure}

    Trivially, the above algorithm runs in polynomial time. Moreover, it returns
    a good grid since every rectangle in $\Dd$ is intersected by some horizontal and some vertical line from $G$. So if $|G|\leq \frac{2k^2}{\eps}$, we can just return $G$ as the output of the algorithm.
    
    It remains to show that if $|G| > \frac{2k^2}{\eps}$, then we can find an independent set of
    weight at least $\opt_k(\Dd)$.
    Assuming that $|G| > \frac{2k^2}{\eps}$, either the vertical or the horizontal run of the greedy algorithm returned more than $\frac{k^2}{\eps}$ lines.
    Without loss of generality assume that the vertical run constructed
    rectangles $R^{\textrm{ver}}_1,\ldots,R^{\textrm{ver}}_m$ for some $m > \frac{k^2}{\eps}$. Observe that
    these rectangles form an independent set, because after iteration $i \in [m]$ all
    rectangles with left side to the left of $\ell_i$ are removed. Since we assumed that the
    ratio between lowest and highest weight of a rectangle in $\Dd$ is at least $\eps/k$, we may estimate the weight of
    $\{R^{\textrm{ver}}_1,\ldots,R^{\textrm{ver}}_m\}$ as follows:
    \begin{displaymath}
        \sum_{i=1}^m \omega(R^{\textrm{ver}}_i) \ge m \cdot \frac{\eps \cdot \omega(R_{\max})}{k} \ge
        k \cdot \omega(R_{\max}) \ge \opt_k(\Dd),
    \end{displaymath}
    where $R_{\max}$ is the rectangle of highest weight in $\Dd$.
    Therefore, the rectangles $R^{\textrm{ver}}_1,\ldots,R^{\textrm{ver}}_m$ form a feasible output for the second point of the lemma statement.
\end{proof}

The first step of the algorithm is to run the procedure of Lemma~\ref{lem:step1}. If this procedure returns an independent set of weight at least $\mathrm{opt}_k(\Dd)$, we just return it as a valid output and terminate the algorithm. Otherwise, from now on we may assume that we have constructed a grid $G$ of size at most $2k^2/\eps$ and this grid is good for $\Ss$.

\subsection{Combinatorial types}

Next, we define the notion of the combinatorial type of a rectangle with respect to a grid. This can be understood as a rough description of the position of the rectangle with respect to the grid.

\begin{definition}[Combinatorial Type]
    Let $G$ be a grid. For an axis-parallel rectangle $R$, we define the {\em{combinatorial type}} $T(R)$ of $R$ with respect to $G$ as
    \begin{displaymath}
       T_G(R)\coloneqq R\cap \pnts(G). 
    \end{displaymath}
    In other words, $T_G(R)$ is the set of grid points of $G$ contained in $R$. For a set $\Ss$ of axis-parallel rectangles, the {\em{combinatorial type}} of $\Ss$ is $T_G(\Ss)$, that is, the image of $\Ss$ under $T_G$. By $\Lambda_k^G$ we denote the set of all possible combinatorial types with respect to $G$ of sets $\Ss$ consisting of at most $k$ axis-parallel rectangles.
\end{definition}

We now observe that if a grid is small, there are only few combinatorial types on it.

\begin{lemma}
    \label{lem:step2}
    For every grid $G$ and positive integer $k$, we have $|\Lambda_k^G|\leq 2^{\Oh(k\log |G|)}$. Moreover, given $G$ and $k$, $\Lambda_k^G$ can be constructed in time $2^{\Oh(k\log |G|)}$.
\end{lemma}
\begin{proof}
    The combinatorial type of any axis-parallel rectangle $R$ can be completely characterized by four lines (or lack thereof) in $G$: the left-most and the right-most vertical line of $G$ intersecting $R$, and the top-most and the bottom-most horizontal line of $G$ intersecting $R$. Hence, there are at most $(|G|+1)^4$ candidates for the combinatorial type of a single rectangle. It follows that the number of combinatorial types of sets of at most $k$ rectangles is bounded by 
    \begin{displaymath}
       1+(|G|+1)^4+(|G|+1)^8+\ldots+(|G|+1)^{8k}\in 2^{\Oh(k\log |G|)}. 
    \end{displaymath}
    To construct $\Lambda_k^G$ in time $2^{\Oh(k\log |G|)}$, just enumerate all possibilities as above.
\end{proof}

The next step of the algorithm is as follows. Recall that we work with a grid $G$ of size at most $2k^2/\eps$ that is good for $\Ss$. By Lemma~\ref{lem:step2}, we can compute $\Lambda_k^G$ in time $2^{\Oh(k\log (k/\eps))}$ and we have $|\Lambda_k^G|\leq 2^{\Oh(k\log (k/\eps))}$. Hence, by paying a $2^{\Oh(k\log (k/\eps))}$ overhead in the time complexity, we can guess $\Tt\coloneqq T_G(\opt_k(G))$, that is, the combinatorial type of the optimum solution. Hence, from now on we assume that the combinatorial type $\Tt$ is fixed. Since $\Ss$ is an independent set and $G$ is good for $\Ss$, we may assume that sets in $\Tt$ are pairwise disjoint and nonempty.

\subsection{Reduction to 2-VCSP}

We say that a rectangle $R \in \Dd$ \emph{matches} a subset of grid points $A\subseteq \pnts(G)$ if $T_G(R)=A$, that is, $R \cap \pnts(G) = A$. By $\Dd_A\subseteq \Dd$ we denote the set of rectangles from $\Dd$ that match $A$.

Based on the combinatorial type $\Tt$ we define an instance $I_\Tt$ of
{\sc{2-VCSP}} as follows. The set of variables is $\Tt$. For every $A\in \Tt$, the domain of $A$ is $\Dd_A\cup \{\bot\}$. That is, the set of rectangles from $\Dd$ that match $A$ plus a special symbol $\bot$ denoting that no rectangle matching $A$ is taken in the solution. Also, for every $A\in \Tt$ we add a unary\footnote{Formally, in the definition of {\sc{2-VCSP}} we allowed only binary constraints, but unary constraints --- constraints involving only one variable --- can be modelled by binary constraints binding a variable with itself.} constraint $c_A$ on $A$ with associated revenue function $f_{c_A}\colon \Dd_A\cup \{\bot\}\to \real$ defined as $f_{c_A}(R)=\omega(R)$ for each $R\in \Dd_A$ and $f_{c_A}(\bot)=0$.


\begin{figure}[ht!]
    \centering
    \includegraphics[width=0.4\textwidth]{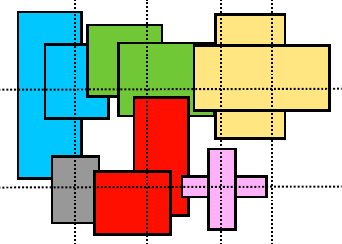}
    \caption{The instance after guessing the combinatorial type $\Tt$. Rectangles that
    match the same type $A\in \Tt$ are filled with the same color. Each
variable corresponds to a different rectangle of the optimum solution, equivalently to a different type $A\in \Tt$, equivalently to a different color in the figure. The domain of a variable consists of all rectangles in the corresponding color. }
    \label{fig:graph}
\end{figure}

It remains to define binary constraints binding pairs of distinct variables in $I_\Tt$.
Two distinct grid points of $G$ are \emph{adjacent} if
they lie on the same or on consecutive horizontal lines of $G$, and on the same or on consecutive vertical lines of $G$.
We put a binary constraint $c_{A,B}$ binding variables $A\in \Tt$ and $B\in \Tt$ if there exist grid points $p \in
A$ and $q \in B$ that are adjacent. The revenue function for $c_{A,B}$ is defined as follows: for $R_A\in \Dd_A\cup \{\bot\}$ and $R_B\in \Tt_B\cup \{\bot\}$, we set
\begin{displaymath}
   f_{c_{A,B}}(R_A,R_B)=\begin{cases}-\infty & \textrm{if }R_A\in \Dd_A\textrm{ and }R_B\in \Dd_B\textrm{ intersect;}\\0 & \textrm{otherwise.}\end{cases}. 
\end{displaymath}
In other words, $c_{A,B}$ is a hard constraint: we require that the rectangles assigned to $A$ and $B$ are disjoint (or one of them is nonexistent), as otherwise the revenue is $-\infty$. This concludes the construction of the instance of $I_\Tt$; clearly, it can be done in polynomial time.


The instance $I_\Tt$ is constructed so that it corresponds to the problem of selecting disjoint rectangles from $\Dd$ that match the combinatorial type $\Tt$. This is formalized in the following statement.

\begin{claim}
    \label{obs:equiv}
    If $\Ss \subseteq \Dd$ is an independent set of rectangles of combinatorial
    type $\Tt$, then there exists a solution to $I_\Tt$ with revenue equal to
    $\omega(\Ss)$. Conversely, if there exists a solution to $I_\Tt$ with
    revenue $r\geq 0$, then there exists an independent set $\Ss\subseteq \Dd$ of weight $r$ and cardinality at most $k$.
\end{claim}

\begin{proof}
 For the first implication, we construct an assignment $\mathbf{u}\colon \Tt\to \Dd$ by setting, for each $A\in \Tt$, $\mathbf{u}(A)$ to be the unique rectangle $R\in \Ss$ for which $T_G(R)=A$. To see that the revenue of $\mathbf{u}$ is equal to $\omega(\Ss)$, note that for every $A\in \Tt$ the unary constraint $c_A$ yields revenue $\omega(\mathbf{u}(A))$, while all binary constraints yield revenue $0$, because the rectangles are pairwise disjoint.

For the second implication, let $\Ss\subseteq \Dd$ be the image of the
assignment $\mathbf{u}$ (possibly with $\bot$ removed). Clearly, $|\Ss|\leq
|\Tt|\leq k$. Since $\mathbf{u}$ yields a nonnegative revenue, all binary
constraints must give revenue~$0$, hence $\omega(\Ss)$ is equal to the revenue
of $\mathbf{u}$, that is, to $r$. It remains to argue that $\Ss$ is an
independent set. For this, take any distinct $A,B\in \Tt$; we need to argue that
in case when rectangles $R_A\coloneqq \mathbf{u}(A)$ and $R_B\coloneqq \mathbf{u}(B)$ are both not equal to $\bot$, they are disjoint. Suppose, for  contradiction, that $R_A$ and $R_B$ have some common point~$x$. Let $Q$ be the cell of the grid $G$ that contains $x$. Since $x\in R_A$ and $A$ is nonempty (recall that this is the assumption about all the sets in $\Tt$, following from $G$ being good for $\opt_k(\Dd)$), $A$ must contain at least one corner of~$Q$, say $p$. Similarly, $B$ contains a corner of~$Q$, say $q$. Note that $p$ and $q$ are adjacent grid points, hence in $I_\Tt$ there is a constraint $c_{A,B}$ that yields revenue $-\infty$ in the case when the rectangles assigned to $A$ and $B$ intersect. As this is the case in~$\mathbf{u}$, we obtain a contradiction with the assumption $r\geq 0$.
\end{proof} 

\subsection{Almost planarity of the Gaifman graph}

Let $H$ be the Gaifman graph of $I_\Tt$; see Figure~\ref{fig:conflict-graph} for an
example. Recall that the vertex set of $H$ is $\Tt$, and distinct $A,B\in \Tt$ are considered adjacent in $H$ iff there is a grid cell $Q$ of $G$ such that both $A$ and $B$ contain a corner of $Q$. Without loss of generality we assume that $H$ is connected, as
otherwise we solve a $I_\Tt$ by treating each connected component
separately and joining the solutions.

\begin{figure}[ht!]
    \centering
    \includegraphics[width=0.4\linewidth]{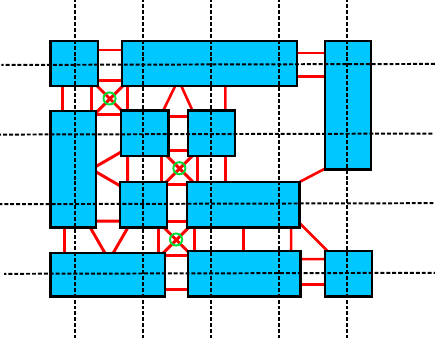}
    \caption{The Gaifman graph $H$ of the constructed {\sc{2-VCSP}}
    instance $I_\Tt$. The vertices are depicted in blue and the edges are
    depicted in thick red. The graph $\Hb$ is constructed from $H$ by introducing
    a new vertex at the intersection of every crossing (hence in the Figure we
    need to add green stroked vertices).}
    \label{fig:conflict-graph}
\end{figure}

%
%

Note that the graph $H$ is not necessarily planar, as there might be crossings within cells; this happens when all four corners belong to different elements of $\Tt$.
However, the intuition is that the crossings within cells are the only problem, hence $H$ is almost planar.  We would
like to apply Baker's technique on $H$. We do it in an essentially direct way, except that we need to be careful about the aforementioned crossings. For this, the following construction will be useful.

 Call a grid cell $Q$ a {\em{cross}} if $Q$ has four corners and all those four corners belong to pairwise different elements of $\Tt$. Note that then all those four elements form a clique in $H$. 
 Construct a graph $\Hb$ from $H$ as follows: add every cross $Q$ to the vertex set, make it adjacent to all four elements of $\Tt$ containing the corners of $Q$, remove the edge connecting the elements of $\Tt$ containing the top-left and the bottom-right corner of $Q$, and to the same for the top-right and bottom-left corners. 
 
 The reader may imagine $\Hb$ as obtained from $H$ by introducing a new vertex at the intersection of diagonals within every cross $Q$; this new vertex is identified with $Q$.
 See Figure~\ref{fig:conflict-graph}. So we have the following simple observation.
 
 \begin{claim}
  The graph $\Hb$ is planar.
 \end{claim}
 \begin{proof}
  Let $H^\bullet_0$ be the graph consisting of the grid points of $G$ where two grid points are adjacent if they are consecutive on the same line of $G$, plus we add a new vertex for every cell of $G$ and make it adjacent to all the corners of this cell. Clearly, $H^\bullet_0$ is planar. Now, $\Hb$ can be obtained from $H^\bullet_0$~as follows:
  \begin{itemize}[nosep]
   \item contract every $A\in \Tt$ to a single vertex;
   \item remove every element of $\pnts(G)\setminus \bigcup \Tt$; and
   \item for every grid cell $Q$ of $G$ that is not a cross, either contract the vertex corresponding to $Q$ onto any of its neighbors, or remove it if it has no neighbors. 
  \end{itemize}
  So $\Hb$ is a minor of a planar graph, hence it is planar as well.
 \end{proof}

 We also have the following simple claim.
 
 \begin{claim}
    \label{fact:dist}
  For all $A,B\in \Tt$, $\dist_\Hb(A,B)\leq 2\cdot \dist_H(A,B)$.
 \end{claim}
 \begin{proof}
  By repeated use of triangle inequality along a shortest path connecting $A$ and $B$, it suffices to argue the following: if $A$ and $B$ are adjacent in $H$, then they are at distance at most $2$ in $\Hb$. For this, observe that either $A$ and $B$ are still adjacent in $\Hb$, or they contain two opposite corners of some cross $Q$, which becomes their common neighbor in $\Hb$. 
 \end{proof}

\newcommand{\Ll}{{\mathcal{L}}}

We now apply Baker's technique.
Select any $A\in \Tt$ and partition $\Tt$ into layers according to the distance in $H$ from $A$: for a nonnegative integer $t$, layer $\Ll_t$ consists of all those vertices $B\in \Tt$ for which $\dist_H(A,B)=t$. Note that layers $\Ll_t$ form a partition of $\Tt$ due to the assumption that $H$ is connected. The following observation is crucial.

 \newcommand{\Ww}{\mathcal{W}}
 \newcommand{\Kk}{\mathcal{K}}

\begin{lemma}\label{lem:tw-bound}
 For all integers $0\leq i<j$, the treewidth of $H[\Ll_i\cup \Ll_{i+1}\cup \ldots\cup \Ll_j]$ is bounded by $\Oh(j-i)$. 
\end{lemma}
\begin{proof}
 We shall assume that $i>0$; at the end we will quickly comment on how the case $i=0$ is resolved in essentially the same way.
 
 Let $\Ww\subseteq V(\Hb)$ be the union of all layers $\Ll_t$ with $i\leq t\leq j$, plus all crosses $Q$ that, in $\Hb$, are adjacent to any element of those layers.
 Further, let $\Kk\subseteq V(\Hb)$ be the union of all layers $\Ll_t$ with $0\leq t< i$, plus all crosses $Q$ that, in $\Hb$, are adjacent to any element of those layers, except those that are already included in $\Ww$. In this way, $\Kk$ and $\Ww$ are disjoint. Further, observe that from the definition of sets $\Ll_t$ as distance layers from $A$ it follows that both $\Hb[\Kk]$ and $\Hb[\Kk\cup \Ww]$ are connected.
 
 Let $H'$ be the graph obtained from $\Hb[\Kk\cup \Ww]$ by contracting the subgraph $\Hb[\Kk]$ into a single vertex; call it $A'$. As a minor of a planar graph, $H'$ is again planar. Next, by the definition of the layers, for every $B\in \Ll_i\cup \ldots \cup \Ll_j$ there exists $C\in \Ll_{i-1}$ such that $\dist_H(C,B)\leq j-i+1$. By Claim~\ref{fact:dist}, we have $\dist_\Hb(C,B)\leq 2(j-i+1)$, implying that $\dist_{H'}(A',B)\leq 2(j-i+1)$. Since every cross included in $\Ww$ is adjacent to some $B$ as above, we conclude that $H'$ is a connected planar graph of radius at most $2(j-i+1)+1=2(j-i)+3$. By standard bounds linking treewidth with radius in planar graphs (see for instance~\cite[2.7]{RobertsonS84}), we conclude that $H'$ has treewidth at most $6(j-i)+10$.
 
 It remains to connect the treewidth of $H'$ with the treewidth of $H''\coloneqq H[\Ll_i\cup \Ll_{i+1}\cup \ldots\cup \Ll_j]$. For this, let $(T,\bag)$ be a tree decomposition of $H'$ of width at most $6(j-i)+10$. We turn $(T,\bag)$ into a tree decomposition $(T,\bag')$ of $H''$ as follows. For every cross $Q\in \Ww$, let $N_Q$ be the set of (at most four) neighbors of $Q$ in $H'$. Then $(T,\bag')$ is obtained by first removing $A'$ from all bags, and then, for every cross $Q\in \Ww$, replacing $Q$ with $N_Q$ in all bags of $(T,\bag)$ that contain $Q$. Since every pair $B,B'\in \Ww$ that is adjacent in $H''$ but not in $H'$ has some cross $Q\in \Ww$ as a common neighbor, and elements of $N_Q$ are adjacent to $Q$ for every cross $Q\in \Ww$, it is straightforward to verify that $(T,\bag')$ is a tree decomposition of $H''$. Finally, in the transformation described above the cardinalities of bags grow by a multiplicative factor of at most $4$, hence the width of $(T,\bag')$ is at most $24(j-i)+43\in \Oh(j-i)$.
 
 This finishes the proof for the case $i>0$. If $i=0$, we perform the same reasoning except that we do not contract $\Kk$ (which now is empty). Thus, we simply work with $H'=\Hb[\Ww]$, and this graph has radius at most $2(j-i)=2j$ by Claim~\ref{fact:dist}. The rest of the reasoning applies verbatim.
\end{proof}

\subsection{Proof of Theorem~\ref{thm:main-rect}}

We are now ready to finish the proof of Theorem~\ref{thm:main-rect}. Recall that the steps performed so far were as follows:
\begin{itemize}[nosep]
 \item We guessed a rectangle $R_{\max}\in \Ss$ (optimum solution) and removed all rectangles of weight larger than $\omega(R_{\max})$ or not exceeding $\eps\cdot \omega(R_{\max})/k$. This induced a loss of at most $\eps\cdot \opt_k(\Dd)$ on the  optimum.
 \item We applied the algorithm of Lemma~\ref{lem:step1}. This way, we either find an independent set with a suitably large weight, or we construct a grid $G$ of size $|G|\leq 2k^2/\eps$.
 \item We used Lemma~\ref{lem:step2} to guess, by branching into $2^{\Oh(k\log (k/\eps))}$ possibilities, the combinatorial type $\Tt$ of an optimum solution.
 \item We constructed a {\sc{2-VCSP}} instance $I_\Tt$ corresponding to the type $\Tt$.
\end{itemize}
By Claim~\ref{obs:equiv}, it remains to find a solution to $I_\Tt$ that yields revenue at least $(1-\eps)\opt(I_\Tt)$, where $\opt(I_\Tt)$ is the maximum possible revenue in $I_\Tt$. (Note that by retracing previous steps, this will result in finding a solution to the original instance of {\sc{MWISR}} of weight at least $(1-2\eps)\opt_k(\Dd)$, so at the end we need to apply the reasoning to $\eps$ scaled by a factor of $1/2$.)

As argued before, we may assume that $H$, the Gaifman graph of $I_\Tt$, is connected. We partition $\Tt$ into layers $\{\Ll_t\colon t=0,1,2,\ldots\}$ as in the previous section. Let $\ell\coloneqq\lceil 1/\eps\rceil$, and define
\begin{displaymath}
   \Mm_r\coloneqq \bigcup_{t\equiv r\bmod \ell} \Ll_t\qquad\textrm{for all }r\in \{0,1,\ldots,\ell-1\}. 
\end{displaymath}
Note that $\{\Mm_r\colon r\in \{0,1,\ldots,\ell-1\}\}$ is a partition of $\Tt$. 

Let $\mathbf{u}$ be an optimum solution to $I_\Tt$. Since it is always possible to assign $\bot$ to every element of $\Tt$, we have $f(\mathbf{u})\geq 0$, in particular all (hard) binary constraints are satisfied under $f$. Therefore, $f(\mathbf{u})=\sum_{r=0}^{\ell-1} f(\mathbf{u}|_{\Mm_r})$. Since $\ell\geq 1/\eps$, there exists $r_0\in \{0,1,\ldots,\ell-1\}$ such that $f(\mathbf{u}|_{\Mm_{r_0}})\leq \eps\cdot f(\mathbf{u})$. The algorithm guesses, by branching into $\ell$ possibilities, the value of $r_0$.

Let $I'_\Tt$ be the {\sc{2-VCSP}} instance obtained from $I_\Tt$ by deleting all variables contained in $\Mm_{r_0}$. Observe that we have $\opt(I'_{\Tt})\geq (1-\eps)\cdot \opt(I_\Tt)$, since $\mathbf{u}$ restricted to the variables of $I'_\Tt$ yields revenue at least $(1-\eps)\cdot \opt(I_\Tt)$. Further, every solution to $I'_{\Tt}$ can be lifted to a solution to $I_\Tt$ of the same revenue by just mapping all variables of $\Mm_{r_0}$ to $\bot$. Hence, it suffices find an optimum solution to $I'_\Tt$.

For this, observe that the Gaifman graph of $I'_\Tt$ is equal to $H-\Mm_{r_0}$. This graph is the disjoint union of several subgraphs, each contained in the union of at most $\ell-1$ consecutive layers $\Ll_t$. By Lemma~\ref{lem:tw-bound} we infer that the treewidth of $H-\Mm_{r_0}$ is bounded by $\Oh(\ell)=\Oh(1/\eps)$. We apply Lemma~\ref{thm:vcsp} to solve $I'_\Tt$ optimally in time $|\Dd|^{\Oh(1/\eps)}\cdot k^{\Oh(1)}$. Together with the previous guessing steps, this gives time complexity $2^{\Oh(k\log (k/\eps))}\cdot |\Dd|^{\Oh(1/\eps)}$ in total, as wanted. This concludes the proof of Theorem~\ref{thm:main-rect}.

\section{Axis-parallel segments}\label{sec:segm}

In this section we prove Theorem~\ref{thm:main-rect}. 
We use the same notation as in Section~\ref{sec:rect}: $\Dd$ is the given set of axis-parallel segments, $\omega\colon \Dd\to \real$ is the weight function on $\Dd$, and $\opt_k(\Dd,\omega)$ is the maximum possible $\omega$-weight of a set of at most $k$ disjoint segments in $\Dd$; we may drop $\omega$ in the notation if the weight function is clear from the context. Also, whenever $\Dd$, $\omega$, and $k$ are clear from the context, then by an {\em{optimum solution}} we mean a set of pairwise disjoint segments $\Ss\subseteq \Dd$ such that $|\Ss|\leq k$ and $\omega(\Ss)=\opt_k(\Dd)$.

\subsection{Reducing the number of distinct weights}\label{sec:weight-reduction}

We first apply the same preprocessing as in Section~\ref{sec:rect}: we guess a segment $R_{\max}\in \Ss$ of maximum weight and remove from $\Dd$ all segments of weight larger than $\omega(R_{\max})$ or not exceeding $\eps\cdot \omega(R_{\max})/k$. Let $\Dd'\subseteq \Dd$ be the obtained set of segments. As argued in Section~\ref{sec:rect}, we have 
\begin{displaymath}
   \opt_k(\Dd',\omega)\geq (1-\eps)\cdot \opt_k(\Dd,\omega). 
\end{displaymath}
As the next preprocessing step, we round all weights down to the set 
\begin{displaymath}
   M\coloneqq \{\omega(R_{\max})(1-\eps)^i\colon i\in \{0,1,\dots,\lceil\log_{1-\eps}(\eps/k)\rceil\}\}. 
\end{displaymath}
That is, we define the new weight function $\omega'\colon \Dd'\to \real$ by setting $\omega'(R)$ to be the largest element of $M$ not exceeding $\omega(R)$. Since the weight of every segment is scaled down by a multiplicative factor of at most $1-\eps$, we have
\begin{displaymath}
    \opt_k(\Dd',\omega')\geq (1-\eps)\cdot \opt_k(\Dd',\omega)\geq (1-\eps)^2\cdot \opt_k(\Dd,\omega)\geq (1-2\eps)\cdot \opt_k(\Dd,\omega).   
\end{displaymath}
Thus, the two preprocessing steps above reduce solving the instance $(\Dd,\omega)$ to solving the instance $(\Dd',\omega')$, at the cost of losing $2\eps\cdot \opt_k(\Dd)$ on the optimum and a $|\Dd|^{\Oh(1)}$ overhead in the time complexity. Observe that in $(\Dd',\omega')$, the number of distinct weights of rectangles is bounded by $|M|\leq \Oh(\eps\log (k/\eps))$. We show in the sequel, that the {\sc{MWISR}} problem for axis-parallel segments can be solved in fixed-parameter time when parameterized by both $k$ and the number of distinct weights.

\begin{theorem}
    \label{thm:segments}
    Suppose $\Dd$ is a set of axis-parallel segments in the plane and $\omega$ is a positive weight function on~$\Dd$. Let $W$ be the number of distinct weights assigned by $\omega$.
Then given~$k$, in time $(kW)^{\Oh(k^2)}\cdot |\Dd|^{\Oh(1)}$ one can find an optimum solution.
\end{theorem}

Note that Theorem~\ref{thm:main-segm} follows from combining Theorem~\ref{thm:segments} with the preprocessing described above (applied to $\eps$ scaled by a factor of $1/2$).
Therefore, from now on we focus on proving
Theorem~\ref{thm:segments}.


\subsection{Constructing a grid}\label{sec:grid}

Let $\preceq$ be any total order on $\Dd$ that orders the segments by non-decreasing weights, that is, $\omega(R)<\omega(R')$ entails $R\prec R'$. Extend $\preceq$ to subsets of $\Dd$ in a lexicographic manner: $\Ss\prec \Ss'$ if $\Ss$ is lexicographically smaller than $\Ss'$ where both sets are sorted according to $\preceq$. 
Let $\Ss_{\max}$ be the $\preceq$-maximum optimum solution.


The next step is to guess (by branching into $|\Dd|$ options) the $\preceq$-minimum segment $R_{\min}$ of $\Ss_{\max}$. Having done this, we safely remove from $\Dd$ all segments $R$ satisfying $R\prec R_{\min}$. Since $\Ss_{\max}$ is the $\preceq$-maximum optimum solution, we achieve the following property: every optimum solution contains the $\preceq$-smallest segment of $\Dd$. We proceed further with this assumption.

We now show that under this assumption, there must exist a grid of size at most $k$ that hits every segment from $\Dd$. Here and later on, a segment is {\em{hit}} by a line if they intersect, and a segment is {\em{hit}} by a grid if it is hit by a line in this grid.

\begin{claim}\label{cl:k-lines}
    Suppose every optimum solution contains the $\preceq$-minimum segment of $\Dd$. Then 
	there exists a grid $G$ of size at most $k$ such that every segment in $\Dd$ is hit by $G$.
\end{claim}
\begin{proof}
    Let $R_0$ be the $\preceq$-minimum segment of $\Dd$ and let $\Ss$ be any optimum solution.
    Let $G$ be the grid comprising of, for every segment $R\in \Ss$, the line containing $R$. Clearly, we have $|G|\leq |\Ss|\leq k$. Suppose for contradiction, that there is a segment $R\in \Dd$ which is not hit by any line of $G$. Clearly $R\neq R_0$, because $R_0\in \Ss$ by assumption.
    Consider $\Ss'\coloneqq \Ss-\{R_0\}\cup \{R\}$ and note that $\Ss'$ is an independent set, because all segments of $\Ss$ are contained in lines of $G$, while $R$ is disjoint with all those lines. Since $R_0\prec R$, we have $\omega(R_0)\leq \omega(R)$, hence $\omega(\Ss')\geq \omega(\Ss)$. So $\Ss'$ is an optimum solution that does not contain $R_0$, a~contradiction.
\end{proof}

Note that the proof of Claim~\ref{cl:k-lines} is non-constructive, as the definition of the grid depends on the (unknown) optimum solution $\Ss$. However, we can give a polynomial-time $\Oh(k)$-approximation algorithm for finding a grid that hits all segments in $\Dd$.

\begin{lemma}\label{lem:k2-lines}
    There exists a polynomial time algorithm that, given a set $\Dd$ of axis-parallel segments in the plane and an integer $k$, either correctly concludes that there is no grid of size at most $k$ which hits all segments of~$\Dd$, or finds such a grid of size $\Oh(k^2)$.
\end{lemma}

\begin{proof}
    We construct a grid $G$ as follows.  Swipe a vertical line from left to
    right across $\Dd$ until the first moment when the segments lying entirely to the left
    of the line can not be hit by $k$ horizontal lines anymore. Let $x_1$ be the position of the line at this moment; in other words, $x_1$ is the least real such that the segments of $\Dd$ entirely contained in $(-\infty,x_1]\times \real$ cannot be hit with $k$ horizontal lines. We set $x_1=\infty$ in case the whole $\Dd$ can be covered with at most $k$ horizontal lines. By the minimality of $x_1$, the segments entirely contained in $(-\infty,x_1]\times \real$ can be covered by $k+1$ lines: the $k$ horizontal lines required to cover segments in $(-\infty,x_1)$, plus one vertical line at $x_1$ (in case $x_1\neq \infty$). We add all those $k+1$ lines to $G$, delete from $\Dd$ all segments hit by those lines, and repeat the procedure until no more segments are left in $\Dd$. This way we obtain numbers $x_1\leq x_2\leq \ldots \leq x_\ell$ and a grid $G$ of size at most $(k+1)\ell$, where $\ell$ is the number of iterations of the procedure. See Figure~\ref{fig:seg-grid} for an
    illustration.

    \begin{figure}[ht!]
        \centering
        \includegraphics[width=0.4\textwidth]{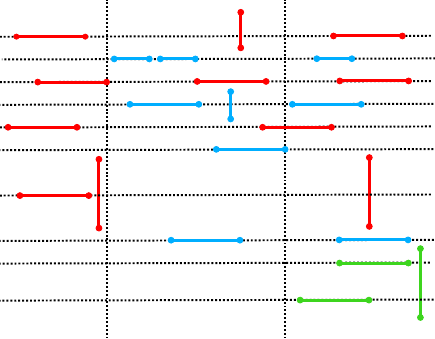}
        \caption{Example of the grid construction for $k=3$. Subsequently, red,
            blue and then green
        segments are removed in consecutive iterations. In each
        iteration we
    scan the segments from left to right until $k+1$ horizontal lines are needed to cover the already seen segments.
    In the last iteration at most $k$ horizontal lines are selected. Lines added to $G$ are dotted.}
        \label{fig:seg-grid}
    \end{figure}
    
    Clearly, $G$ hits all segments in $\Dd$. So if $\ell\leq k+1$, then $|G|\leq (k+1)^2=\Oh(k^2)$ and the algorithm can provide $G$ as a valid output. We now argue that if $\ell>k+1$, then the algorithm may safely conclude that there is no grid of size at most $k$ that hits all segments of $\Dd$. For contradiction, suppose there is such a grid $G'$. For $i\in \{1,\ldots,\ell\}$, let $\Dd_i$ be the set of all segments entirely contained in $(x_{i-1},x_i]\times \real$, where we set $x_0=-\infty$. It is easy to see that $\Dd_i$ is precisely the set of segments for which the algorithm in iteration $i$ decided that it cannot be hit by at most $k$ horizontal lines. Hence, for each $i\in \{1,\ldots,\ell\}$, $G'$ must contain at least one vertical line hitting at least one segment in $\Dd_i$. The $x$-coordinate of this vertical line must belong to the interval $(x_{i-1},x_{i}]$, so these vertical lines must be pairwise different. We conclude that $|G'|\geq \ell>k$, a contradiction.

	
    It remains to argue how to implement the algorithm so that it runs in polynomial time. Observe that for a set of segments $\Dd'\subseteq \Dd$, the minimum number of horizontal lines needed to hit all the segments of $\Dd'$ can be computed as follows: Projecting all the segments $\Dd'$ on the vertical axis, and find the minimum number of points that hit the obtained set of intervals (some of which are single points; these are projected horizontal segments). This, in turn, can be done in time $\Oh(|\Dd'|\log |\Dd'|)$ using a standard greedy strategy. It is now straightforward to use this sub-procedure to execute the construction of $G$ described above in polynomial~time.
\end{proof}

We now combine Claim~\ref{cl:k-lines} and Lemma~\ref{lem:k2-lines} as follows. Run the algorithm of Lemma~\ref{lem:k2-lines} on $\Dd$ with parameter $k$. If the algorithm concludes that there is no grid of size at most $k$ that hits all segments of $\Dd$, then by Claim~\ref{cl:k-lines} we can terminate the current branch, as clearly one of the previous guesses was incorrect. Otherwise, we obtain a grid $G$ of size $\Oh(k^2)$ that hits every segment of $\Dd$. With this grid we proceed to the next steps.

For brevity of presentation, by adding four lines to $G$ we may assume that all segments of $\Dd$ are contained in the interior of the rectangle delimited by the left-most and the right-most vertical line of $G$ and the top-most and the bottom-most horizontal line of $G$. We will also say that a grid with this property {\em{encloses}} $\Dd$.

\subsection{Constructing a nice grid}

We use the same notion of niceness as in Section~\ref{sec:rect}. That is, a grid $G$ is \emph{nice} with respect to a segment~$R$, if
$R$ contains at least one grid point of $G$; in other words, $R$ is intersected by both a horizontal and a vertical line in $G$. 
We will also say that $R$
\emph{respects} the grid $G$.
The {\em{ugliness}} of a grid $G$ with respect to some optimum solution $\Ss$ is the number of segments of $\Ss$ that do not respect $G$. Then the {\em{ugliness}} of $G$ is the minimum over all optimum solutions $\Ss$ of the ugliness of $G$ with respect to $\Ss$. This way, a grid is {\em{nice}} if its ugliness is~$0$, or equivalently, there exists an optimum solution $\Ss$ such that $G$ is nice with respect to all the segments in $\Ss$.

In further considerations, it will be convenient to again rely on a suitable
defined notion of a combinatorial type of a segment with respect to a grid.
Consider a grid $G$ that encloses $\Dd$.
For a segment $R\in \Dd$, the {\em{combinatorial type}} of $R$ with respect to~$G$ is the $6$-tuple consisting of:
\begin{itemize}[nosep]
     \item The boolean value indicating whether $R$ is horizontal or vertical.
     \item The weight $\omega(R)$.
     \item The right-most line $\ell_{\leftarrow}$ of $G$ such that $R$ entirely lies strictly to the right of $\ell_{\leftarrow}$.
     \item The left-most line $\ell_{\rightarrow}$ of $G$ such that $R$ entirely lies strictly to the left of $\ell_{\rightarrow}$. 
     \item The bottom-most line $\ell_{\uparrow}$ of $G$ such that $R$ entirely lies strictly below $\ell_{\uparrow}$.
     \item The top-most line $\ell_{\downarrow}$ of $G$ such that $R$ entirely lies strictly above $\ell_{\downarrow}$. 
    \end{itemize}
In other words, $(\ell_{\leftarrow},\ell_{\rightarrow},\ell_{\uparrow},\ell_{\downarrow})$ contain the sides of the inclusion-wise minimal rectangle $R'$ delimited by the lines from $G$ whose interior contains $R$. Note that the set of grid points of $G$ contained in $R$ is equal to the set of grid points contained in the interior of $R'$. Assuming $G$ is clear from the context, for a type $t$ we will denote this set of grid points by $P(t)$. Observe that the number of different combinatorial types with respect to~$G$ is bounded by $2W|G|^4$, where $W$ is the number of distinct weights assigned by $\omega$.

\begin{lemma}\label{lem:nice-suffices}
 Given a finite set $\Dd$ of axis-parallel segments in the plane,  a positive weight function $\omega$ on $\Dd$, a positive integer $k$, and a grid $G$ that encloses $\Dd$ with a guarantee that the ugliness of $G$ is~$0$. Then an optimum solution for $\Dd,\omega,k$ can be found in time $(W\cdot |G|)^{\Oh(k)}\cdot |\Dd|^{\Oh(1)}$.
\end{lemma}
\begin{proof}
 Fix any optimum solution $\Ss$ such that $G$ is nice with respect to $\Ss$. We guess, by branching into all possibilities, the combinatorial types (with respect to $G$) of all the segments of $\Ss$. Since there are at most $2W|G|^4$ different combinatorial types, this results in $(W\cdot |G|)^{\Oh(k)}$ branches. Let the guessed set of combinatorial types be $\Tt$. Since $G$ is supposed to be nice with respect to $\Ss$, we may assume that the sets $\{P(t)\colon t\in \Tt\}$ are nonempty and pairwise disjoint; otherwise the branch can be discarded. 
 
 We construct an auxiliary {\sc{$2$-CSP}} instance $I$ that models the choice of
 segments in $\Ss$. The set of variables is $\Tt$. For every type $t\in \Tt$,
 the domain $\Dd_t$ consists of all segments from $\Dd$ whose combinatorial type is~$t$. The constraints are as follows: 
    \begin{itemize}[nosep]
     \item If $t,t'\in \Tt$ are distinct types of horizontal segments, and $P(t)$ and $P(t')$ are two adjacent intervals of grid points on the same horizontal line of $G$, then we put a constraint between $t$ and $t'$ that among $\Dd_t\times \Dd_{t'}$, allows only pairs of disjoint segments.
     \item Analogous constraints are put for distinct types $t,t'\in \Tt$ of vertical segments for which $P(t)$ and $P(t')$ are adjacent intervals on the same vertical line.
    \end{itemize}
    It is straightforward to verify that solutions to $I$ correspond in one-to-one fashion to those independent sets in $\Dd$ for which the set of combinatorial types is $\Tt$. Moreover, observe that the Gaifman graph of $I$ is a disjoint union of paths, where every path $t_1-\ldots-t_p$ corresponds to a sequence $P(t_1),\ldots,P(t_p)$ of intervals on the same grid line such that $P(t_i)$ is adjacent to $P(t_{i+1})$ for $i\in \{1,\ldots,p-1\}$. Therefore, it suffices to solve $I$ optimally, which can be done in time $|\Dd|^{\Oh(1)}$ using, for instance, Lemma~\ref{thm:vcsp}.
\end{proof}

Lemma~\ref{lem:nice-suffices} suggests that we should aim to construct a grid with zero ugliness. 
So far, the grid $G$ constructed in the previous section may have positive ugliness: some segments of $\Dd$ may be intersected by just one, and not two orthogonal lines, and there is no reason why an optimum solution should not contain any such segments. Our goal is to reduce the ugliness of the grid by further branching steps. The branching strategy is captured in the following lemma.


\begin{lemma}\label{lem:branching}
    Given a finite set $\Dd$ of axis-parallel segments in the
    plane,  a positive weight function $\omega$ on $\Dd$, a positive integer
    $k$, and a grid $G$ that hits all segments of $\Dd$ and encloses $\Dd$.
    Let $W$ be the number of different weights assigned by $\omega$. Then one
    can construct, in time $(|G|\cdot W)^{\Oh(k)}\cdot |\Dd|^{\Oh(1)}$, a family $\Gg$ of grids with the following properties:
    \begin{enumerate}[align=left, font=\normalfont, label=(\roman*),nosep]
     \item\label{p:size} $|\Gg|\leq (|G|\cdot W)^{\Oh(k)}$;
     \item\label{p:inc} for each $G'\in \Gg$, we have $G'\supseteq G$ and $|G'\setminus G|\leq k$; and
     \item\label{p:ugly} If the ugliness of $G$ is positive, then there is $G'\in \Gg$ whose ugliness is strictly smaller than that of $G$. 
    \end{enumerate}
\end{lemma}


\begin{proof}
    Fix an optimum solution $\Ss$ such that the ugliness of $G$ with respect to $\Ss$ is minimum possible. Assume that this ugliness is positive, since otherwise \ref{p:ugly} holds vacuously and any family $\Gg$ satisfying \ref{p:size} and~\ref{p:inc} is a valid output (this will be guaranteed by the algorithm).
    We construct $\Gg$ by a branching algorithm that, intuitively, guesses a bounded amount of information about $\Ss$ and augments $G$ according to the guess, so that the augmented grid is nice with respect to at least one more segment of $\Ss$. Thus, different members of $\Gg$ correspond to different guesses on the structure of $\Ss$.
    
    Let
    $\Nn$ be the set of all segments in $\Ss$ that respect $G$. As the ugliness of $G$ is positive, $\Ss\setminus \Nn$ is nonempty. Let $R_{\max}$ be the maximum weight segment of $\Ss\setminus \Nn$; in case there are several with the same maximum weight, pick any of them.
    
    The algorithm guesses, by branching into all possibilities, the combinatorial types of all segments in $\Nn\cup \{R_{\max}\}$; this results in at most $(1+2W|G|^4)^k\leq  (|G|\cdot W)^{\Oh(k)}$ branches. For every guess we shall construct one grid $G'\supseteq G$ included in $\Gg$. Therefore, we fix one guess and proceed to the description of $G'$. 
    
    By symmetry, we may assume that $R_{\max}$ is horizontal.
    Let $\Tt$ be the (already guessed) set of combinatorial types of segments from $\Nn$, and let $t_{\max}=(\mathsf{horizontal},w,\ell_{\leftarrow},\ell_{\rightarrow},\ell_{\uparrow},\ell_{\downarrow})$ be the (already guessed) combinatorial type of $R_{\max}$. Similar to the proof of Lemma~\ref{lem:nice-suffices}, we assume that the sets $\{P(t)\colon t\in \Tt\}$ are nonempty and pairwise disjoint, as otherwise the guess can be safely discarded as incorrect. Also, note that each $P(t)$ is an interval consisting of consecutive grid points on a single line of $G$.
    
    
    Let $B$ be the rectangle delimited by $(\ell_{\leftarrow},\ell_{\rightarrow},\ell_{\uparrow},\ell_{\downarrow})$.
    Since $G$ is not nice with respect to $R_{\max}$, the interior of $B$ does not contain any grid point of $G$. So there are two cases to consider:
    \begin{itemize}[nosep]
     \item[{\bf{Case 1:}}] $\ell_\uparrow$ and $\ell_\downarrow$ are consecutive horizontal lines of $G$. This is equivalent to $R_{\max}$ lying in the interior of the horizontal strip between $\ell_\uparrow$ and $\ell_\downarrow$. In particular $R_{\max}$ is not contained in any line of $G$. Note that since $R_{\max}$ is hit by  $G$ (which is true about every segment of $\Dd$), the two vertical lines $\ell_\leftarrow$ and $\ell_\rightarrow$ are non-consecutive in $G$. So $B$ is the union of two or more horizontally adjacent grid cells of~$G$.
     \item[{\bf{Case 2:}}] $\ell_\uparrow$ and $\ell_\downarrow$ are non-consecutive horizontal lines of $G$. Since $R_{\max}$ is horizontal, there must exist exactly one line of $G$ between  $\ell_\uparrow$ and $\ell_\downarrow$, say $\ell$, and $\ell$ must contain $R_{\max}$. Note that since $R_{\max}$ contains no grid point of $G$, $\ell_\leftarrow$ and $\ell_\rightarrow$ must be two consecutive vertical lines of $G$. So $B$ is the union of two vertically adjacent grid cells of $G$.
    \end{itemize}
    We consider these two cases separately. See Figure~\ref{fig:case} for an illustration.

%
%
%
    \begin{figure}
        \centering
        \begin{subfigure}{0.45\textwidth}
            \includegraphics[width=\textwidth]{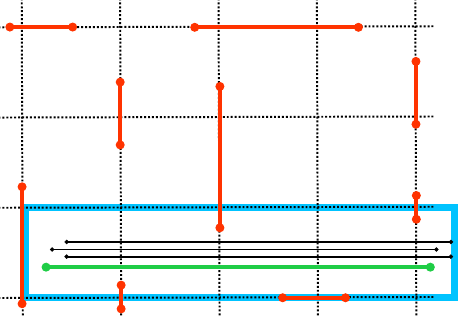}
        \end{subfigure}
        \hfill
        \begin{subfigure}{0.45\textwidth}
            \includegraphics[width=\textwidth]{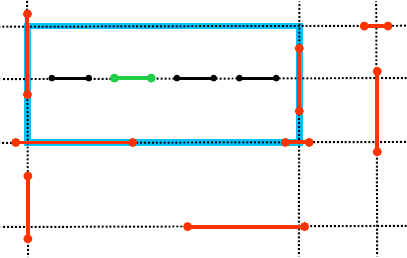}
        \end{subfigure}

        \caption{Illustration of Case 1 and Case 2 of the proof of
        Lemma~\ref{lem:branching}. Red segments are the segments from $\Nn$ (they already contain a grid point). The green segment is the candidate segment
$R_{\max}$ (maximum weight segment in the optimum solution not containing a grid
point). The box $B$ is the blue-stroked rectangle. We greedily find a
maximum-size region inside $B$ containing candidate segments, depicted in black, with the same
combinatorial type as $R_{\max}$. If there are
more than $k$ independent candidates, we can return an optimum solution
(since $R_{\max}$ has maximum weight). Otherwise we can add fewer than $k$ grid lines to the current grid (such that each candidate is hit by a newly added grid line).}
        \label{fig:case}
    \end{figure}
    
    \paragraph*{Case 1: \normalfont{$R_{\max}$ does not lie on a grid line.}} We construct an auxiliary {\sc{$2$-CSP}} instance $I$ that corresponds to the choice of segments in $\Nn$, exactly as in the proof of Lemma~\ref{lem:nice-suffices}. That is, the set of variables is $\Tt$, and the constraints are as described in the proof of Lemma~\ref{lem:nice-suffices}. Again, solutions to $I$ are in one-to-one correspondence to those independent sets in $\Dd$ whose set of combinatorial types is $\Tt$. Also, the Gaifman graph $H$ of $I$ is a disjoint union of paths, where each path corresponds to a sequence of adjacent intervals of grid points contained in a single grid line of $G$.
    
    The idea is to compute a solution $\mathbf{u}$ to $I$ that leaves ``the most space'' for the placement of $R_{\max}$. For this, for every connected component $C$ of $H$, we do the following. Recall that $C$ is a path, and enumerate the consecutive variables on $C$ as $t_1,\ldots,t_p$. Let $P(C)\coloneqq \bigcup_{i=1}^p P(t_i)$; then $P(C)$ is an interval of grid points on one line of $G$, say $\ell$. We again consider two cases.
    
    First, if $\ell$ is horizontal, or $P(C)$ does not contain any grid point lying in the interior of a side of $B$; compute any solution within $C$, say using the algorithm of Lemma~\ref{thm:vcsp}.
    
    Second, if $\ell$ is vertical and $P(C)$ contains some grid point lying in the interior of a side of $B$, we do as follows. Note that the intersection of $\ell$ with $B$ is a segment. Let $x_\uparrow\in \ell_\uparrow$ and $x_\downarrow\in \ell_\downarrow$ be the endpoints of this segment. Then $x_\uparrow$ and $x_\downarrow$ are two vertically adjacent grid points of $G$ that lie in the interior of the top and the bottom side of $B$, while $P(C)$ contains one or both of $x_\uparrow$ and $x_\downarrow$. For concreteness, assume for now that $P(C)$ contains both $x_\uparrow$ and $x_\downarrow$; the other cases are simpler and will be discussed later. Assume that there is no $i\in \{1,\ldots,p\}$ such that $P(t_i)$ contains both $x_\uparrow$ and $x_\downarrow$, because then the corresponding segment of $\Nn$ would necessarily intersect $R_{\max}$; so if this occurs, we can discard the branch as incorrect. So, up to reversing indexing if necessary, there exists $i\in \{1,\ldots,p-1\}$ such that $x_{\uparrow}\in P(t_i)$ and $x_\downarrow\in P(t_{i+1})$. We compute a solution within $C$ greedily as follows:
    \begin{itemize}[nosep]
    \item First, process variables $t_1,\ldots,t_i$ in this order. When considering $t_j$, assign the segment whose lower endpoint is the highest possible among the available segments of $\Dd_{t_j}$ (that is, disjoint with the segment assigned to $t_{j-1}$, for $j>1$).
    \item Second, apply a symmetric greedy procedure to variables $t_p,t_{p-1},\ldots,t_{i+1}$ in this order, always picking an available segment with the lowest possible higher endpoint.
    \end{itemize}
    In case any of $x_{\uparrow}$ or $x_{\downarrow}$ does not belong to $P(C)$, only one of the above greedy procedures is applied.
    
    If $I$ has a solution, the algorithm described above clearly succeeds in finding some solution $\mathbf{u}$ to $I$. Since we assume $I$ to have a solution --- witnessed by $\Nn$ --- we may terminate the branch as incorrect in case no solution to $I$ is found by the algorithm. Let $\Nn'=\mathbf{u}(\Tt)$ be the independent set of segments found by the algorithm above. It is straightforward to see that the greedy choice of solutions within the components of $H$ justifies the following claim.
    
    \begin{claim}\label{cl:greedy-outcome}
     It holds that $\mathrm{int} B\cap \bigcup \Nn' \subseteq \mathrm{int} B\cap \bigcup \Nn$, where $\mathrm{int} B$ denotes the interior of $B$. Consequently, $R_{\max}$ is disjoint with every segment in $\Nn'$.
    \end{claim}
    
    Now, let $\Rr\subseteq \Dd$ be the set of all segments in $\Dd$ whose combinatorial type is $t_{\max}$ and that are disjoint with all segments in $\Nn'$. By Claim~\ref{cl:greedy-outcome}, we necessarily have $R_{\max}\in \Rr$. Let $L$ be the set comprising of all (horizontal) lines containing some segment $R\in \Rr$. We consider two cases.
    \begin{itemize}[nosep]
     \item If $|L|<k-|\Nn|$, then we add the grid $G'\coloneqq G\cup L$ to $\Gg$. 
     \item If $|L|\geq k-|\Nn|$, then we add the grid $G'\coloneqq G\cup L'$ to $\Gg$, where $L'$ is any subset of $L$ of size $k-|\Nn|$. 
    \end{itemize}
    It remains to argue that in both cases, the ugliness of $G'$ is strictly smaller than that of $G$.
    
    In the case $|L|<k-|\Nn|$, it suffices to note that since $R_{\max}$ is contained in some line of $L$, the grid $G'=G\cup L$ is nice with respect $R_{\max}$, while $G$ is not nice with respect to $R_{\max}$ by assumption.
    
    Consider now the case $|L|\geq k-|\Nn|$. For every line $\ell\in L'$, pick any segment $R_\ell\in \Rr$ that lies on $\ell$. Let $\Ll\coloneqq \{R_{\ell}\colon \ell\in L'\}$. Note that the segments of $\Ll$ are pairwise disjoint due to lying on different horizontal lines, and they are also disjoint from all the segments of $\Nn'$ by the definition of $\Rr$. So $\Nn'\cup \Ll$ is an independent set of segments, and has size $k$. Furthermore, since the combinatorial type also features the weight of a segment, and $R_{\max}$ was chosen to be the heaviest segment within $\Ss\setminus \Nn$, we have $\omega(\Nn')=\omega(\Nn)$ and $\omega(\Ll)\geq \omega(\Ss\setminus \Nn)$. It follows that $\omega(\Nn'\cup \Ll)\geq \omega(\Ss)$, hence $\Nn'\cup \Ll$ is also an optimum solution. But $G'=G\cup L'$ is nice with respect to all the segments of $\Nn'\cup \Ll$, so the ugliness of $G'$ is $0$.

    \paragraph*{Case 2: {\normalfont $R$ lies on a grid line.}}
	This case works in a very similar fashion as the previous one, hence we only outline the differences here. 
	
	Recall that in this case $B$ consists of two vertically adjacent cells of $G$. Let $S$ be the common side of those cells; then our guess on the combinatorial type $t_{\max}$ of $R_{\max}$ says that $R_{\max}$ should be contained in the interior of $S$.
	
	We construct an instance $I$ of {\sc{$2$-CSP}} in exactly the same manner as in Case~1. We solve it using a similar greedy procedure, so that the space left for placing $R_{\max}$ within the interior of $s$ is maximized. Here, there will be at most one connected component of the Gaifman graph of $I$ where a greedy strategy is applied; this is the horizontal component $C$ such that $P(C)$ contains one or both endpoints of $S$. Let $\Nn'$ be the obtained solution to $I$. The analogue of Claim~\ref{cl:greedy-outcome} now says the following: $\mathrm{int} S\cap \bigcup \Nn' \subseteq \mathrm{int} S\cap \bigcup \Nn$, hence $R_{\max}$ is disjoint with every segment in $\Nn'$. Consequently, if we denote $S'\coloneqq \mathrm{int} S\setminus \bigcup \Nn'$, then $S'$ is an open segment that contains $R_{\max}$.
	
	Now, let $\Rr$ be the set of all segments from $\Dd$ contained in $S'$ and whose weight is equal to the guessed weight of $R_{\max}$. Since all segments of $\Rr$ lie on the same line, using a polynomial-time left-to-right greedy sweep we may find a maximum independent set of segments within $\Rr$; call it $\Ll$. Let $L$ be the set of vertical lines passing through the right endpoints of the segments in $\Ll$. Note that by construction of $\Ll$, $L$ hits all segments in $\Ll$. We again consider two subcases:
	\begin{itemize}[nosep]
	 \item If $|\Ll|=|L|<k-|\Nn|$, then we add the grid $G'=G\cup L$ to $\Gg$.
	 \item If $|\Ll|=|L|\geq k-|\Nn|$, then we add the grid $G'=G\cup L'$ to $\Gg$, where $L'$ is any, arbitrarily chosen, subset of $L$ with size $k-|\Nn|$.
	\end{itemize}
    A reasoning analogous to Case~1 shows the following. In the first subcase, $G'$ is nice with respect to $R_{\max}$, hence the ugliness of $G'$ is strictly smaller than that of $G$. In the second case, $\Nn'\cup \Ll$ is an optimum solution and $G'$ is nice with respect to $\Nn'\cup \Ll$, hence the ugliness of $G'$ is~$0$.

%
%
%

\bigskip

In both Case~1 and Case~2 we constructed a grid $G'\supseteq G$ with $|G'\setminus G|\leq k$ whose ugliness is strictly smaller than that of $G$. We conclude the proof by taking $\Gg$ to be the set of all grids $G'$ constructed in this manner.  
\end{proof}

Finally, Lemma~\ref{lem:branching} can be applied in a recursive manner to obtain a nice grid.

\begin{lemma}\label{lem:recursion}
     Given a finite set $\Dd$ of axis-parallel segments in the plane,  a positive weight function $\omega$ on $\Dd$, a positive integer $k$, and a grid $G$ that hits all segments of $\Dd$ and encloses $\Dd$. Let $W$ be the number of different weights assigned by $\omega$. Then one can in time $(k\cdot W\cdot |G|)^{\Oh(k^2)}\cdot |\Dd|^{\Oh(1)}$ construct a family $\Gg$ of grids such that:
\begin{enumerate}[align=left, font=\normalfont, label=(\roman*),nosep]
     \item\label{p:size2} $|\Gg|\leq (k\cdot W\cdot |G|)^{\Oh(k^2)}$;
     \item\label{p:inc2} for each $G'\in \Gg$, we have $G'\supseteq G$ and $|G'\setminus G|\leq k^2$; and
     \item\label{p:ugly2} $\Gg$ contains at least one grid of ugliness $0$.
    \end{enumerate}
\end{lemma}
\begin{proof}
 Starting with $\Gg_0\coloneqq \{G\}$, we iteratively construct families of grids $\Gg_1,\Gg_2,\ldots,\Gg_k$ as follows: to construct $\Gg_i$ from $\Gg_{i-1}$, replace each grid $G\in \Gg_{i-1}$ with the family $\Gg(G)$ obtained by applying Lemma~\ref{lem:branching} to $G$.
 A straightforward induction using properties~\ref{p:size} and~\ref{p:inc} of Lemma~\ref{lem:branching} shows that: $|\Gg_i|\leq (k\cdot W\cdot |G|)^{\Oh(ik)}$, for each $G'\in \Gg_i$ it holds that $G'\supseteq G$ and $|G'\setminus G|\leq ik$, and the construction of $\Gg_i$ takes $(k\cdot W\cdot |G|)^{\Oh(ik)}\cdot |\Dd|^{\Oh(1)}$ time. Moreover, by property~\ref{p:ugly} of Lemma~\ref{lem:branching}, if the minimum ugliness among grids in $\Gg_{i-1}$ is positive, then the minimum ugliness among the grids in $\Gg_i$ is strictly smaller than that in $\Gg_{i-1}$. Since the ugliness of $G$ is at most $k$, it follows that $\Gg\coloneqq \Gg_k$ satisfies all the required properties.
\end{proof}

\subsection{Proof of Theorem~\ref{thm:segments}}

We are ready to assemble all the tools and prove Theorem~\ref{thm:segments}.
%
%
%

\begin{proof}[Proof of Theorem~\ref{thm:segments}]
    As discussed in Section~\ref{sec:grid}, by preprocessing the instance and branching into $\Oh(|\Dd|)$ possibilities, we may assume that we constructed a grid $G$ of size $\Oh(k^2)$ such that every segment in $\Dd$ is hit by $G$. Adding four lines to $G$ ensures that $G$ encloses $\Dd$. Then we apply Lemma~\ref{lem:recursion} to $G$, and we construct a family of grids~$\Gg$ that features at least one grid with ugliness $0$. It now remains to apply Lemma~\ref{lem:nice-suffices} to each grid in $\Gg$ and output the heaviest of the obtained solutions.
    Following directly from the guarantees provided by Lemmas~\ref{lem:nice-suffices} and~\ref{lem:recursion}, this algorithm runs in time $(kW)^{\Oh(k^2)}\cdot |\Dd|^{\Oh(1)}$.
\end{proof}

As argued in Section~\ref{sec:weight-reduction}, Theorem~\ref{thm:main-segm} follows from Theorem~\ref{thm:segments}.

\subparagraph*{Acknowledgements.} The results presented in this paper were obtained
during the trimester on Discrete Optimization at the Hausdorff Research
Institute for Mathematics (HIM) in Bonn, Germany.

\bibliographystyle{abbrv}
\bibliography{bib}



\end{document}